\documentclass[reqno]{amsart}
\usepackage{amsmath, amsthm, amssymb}
\usepackage{color}
\usepackage{amsfonts}
\usepackage{geometry}
\usepackage{tikz}
\usepackage{subfigure}
\usepackage[all]{xy}
\usepackage{enumitem}

\definecolor{myurlcolor}{rgb}{0,0,0.7}
\definecolor{myrefcolor}{rgb}{0.8,0,0}
\usepackage{hyperref}
\hypersetup{colorlinks,
linkcolor=myrefcolor,
citecolor=myurlcolor,
urlcolor=myurlcolor}

\newcommand{\R}{\mathbb{R}}
\newcommand{\N}{\mathbb{N}}
\newcommand{\C}{\mathbb{C}}
\newcommand{\Z}{\mathbb{Z}}
\newcommand{\B}{\mathcal{B}}
\renewcommand{\H}{\mathcal{H}}

\newtheorem{prop}{Proposition}[section]

\newtheorem{cor}[prop]{Corollary}

\newtheorem{definition}[prop]{Definition}
\newtheorem{defn}[prop]{Definition}

\newtheorem{lem}[prop]{Lemma}

\theoremstyle{definition} 
\newtheorem{ex}[prop]{Example}
\newtheorem{prob}[prop]{Problem}
\newtheorem{rem}[prop]{Remark}

\newcommand{\beq}{\begin{equation}}
\newcommand{\eeq}{\end{equation}}
\newcommand{\bea}[1]{\begin{equation}\begin{array}{#1}}
\newcommand{\eea}{\end{array}\end{equation}}
\newcommand{\beqn}{\begin{eqnarray}}
\newcommand{\eeqn}{\end{eqnarray}}

\renewcommand{\rho}{\varrho}

\begin{document}

\title[Entropic Inequalities]{Entropic Inequalities and Marginal Problems}

\author{Tobias Fritz}
\address{ICFO--Institut de Ci\`encies Fot\`oniques,
Mediterranean Technology Park, 08860 Castelldefels (Barcelona),
Spain, and Perimeter Institute for Theoretical Physics, Waterloo, Ontario, Canada}
\email{tfritz@perimeterinstitute.ca}
\author{Rafael Chaves}
\address{ICFO--Institut de Ci\`encies Fot\`oniques,
Mediteranean Technology Park, 08860 Castelldefels (Barcelona), Spain, and Institute for Physics, University of Freiburg, Rheinstrasse 10, D-79104 Freiburg, Germany}
\email{rafael.chaves@icfo.es}

\thanks{We would like to thank Daniel Polani for helpful correspondence and Jonatan Bohr Brask for help with part of the \texttt{MATHEMATICA} code.
This work is supported by the EU STREP QCS, the QESSENCE project, the German Science Foundation (grant CH 843/2–1), the
Excellence Initiative of the German Federal and State Governments (grant ZUK 43) and by Perimeter Institute for Theoretical Physics. Research at Perimeter Institute is supported by the Government of Canada through Industry Canada and by the Province of Ontario through the Ministry of Economic Development and Innovation.}


\begin{abstract}
A marginal problem asks whether a given family of marginal distributions for some set of random variables arises from some joint distribution of these variables. Here we point out that the existence of such a joint distribution imposes non-trivial conditions already on the level of Shannon entropies of the given marginals. These entropic inequalities are necessary (but not sufficient) criteria for the existence of a joint distribution. For every marginal problem, a list of such Shannon-type entropic inequalities can be calculated by Fourier-Motzkin elimination, and we offer a software interface to a Fourier-Motzkin solver for doing so. For the case that the hypergraph of given marginals is a cycle graph, we provide a complete analytic solution to the problem of classifying all relevant entropic inequalities, and use this result to bound the decay of correlations in stochastic processes. Furthermore, we show that Shannon-type inequalities for differential entropies are not relevant for continuous-variable marginal problems; non-Shannon-type inequalities are, both in the discrete and in the continuous case. In contrast to other approaches, our general framework easily adapts to situations where one has additional (conditional) independence requirements on the joint distribution, as in the case of graphical models. We end with a list of open problems.

A complementary article discusses applications to quantum nonlocality and contextuality.
\end{abstract}

\maketitle
\section{Introduction}
\label{introduction}

This work concerns two lines of research which we would like to introduce separately and relate to each other afterwards.

\subsection*{Marginal problems.} Imagine you have three coins $A_1$, $A_2$ and $A_3$. However, for some physical reason, you can only flip two of them at a time. Upon flipping $A_1$ and $A_2$ together, you find that these two coins always give the same outcome: two heads occur with a relative frequency of $\tfrac{1}{2}$ and two tails occur with a relative frequency of $\tfrac{1}{2}$. Upon flipping $A_2$ and $A_3$ together, the same behavior ensues. However upon flipping $A_1$ and $A_3$ together, you find the exactly opposite result, so that the two outcomes are never identical: the two combinations of one head and one tail occur with probability $\tfrac{1}{2}$ each.

Now what will happen when you only flip a single coin? Clearly, all three pairwise combinations are consistent in the sense that they predict each coin to yield heads and tails with relative frequency $\tfrac{1}{2}$ each. Therefore, this has to be the resulting outcome distribution of flipping only one coin by itself.

But what would happen if you were able to flip all three coins at once? Since $A_1$ and $A_2$ are perfectly correlated, and also $A_2$ and $A_3$ are perfectly correlated, it follows that $A_1$ and $A_3$ should also be perfectly correlated. This contradicts the observation that $A_1$ and $A_3$ are perfectly anticorrelated. Therefore, no three-variable joint distribution compatible with the given marginals exists, due to the transitivity of perfect correlation. Although the two-coin outcome distributions give consistent single-coin distributions, no three-coin distribution is compatible with the given data! Hence either there is some systematic error in the coin flips, or the coins are operated by some weird mechanism creating the required distribution as a function depending on which two are flipped together.

This is the simplest non-trivial example of a \emph{marginal problem}: given a list of joint distributions of certain subsets of random variables $A_1,\ldots,A_n$, is it possible to find a joint distribution for all these variables, such that this distribution marginalizes to the given ones? One obvious necessary condition is that for any two of the given distributions which can be marginalized to the same subset of variables, the resulting marginals should be the same. In the ``three coins'' example, we found this to be the case, since the single-coin marginals were unambiguous: each coin by itself is unbiased, and this does not depend on which other coin it is tossed together with. The example also shows that this consistency condition is not sufficient to guarantee the existence of a joint distribution.

Marginal problems naturally arise in several different fields. To us, the most familiar one is ``quantum nonlocality''~\cite{Bell,Fine}, which features close relatives of our unextendability example---with four ``coins'' instead of three and with slightly different given marginal two-coin distributions. In this case, the unextendability has actually been observed experimentally~\cite{Aspect}, bearing witness to the counterintuitive behavior of quantum theory. As part of the endeavor to understand the counterintuitive features of quantum theory, marginal problems have become an active field of research within the foundations of quantum mechanics~\cite{AbramBrand,AQ,CSW,LSW}. Unfortunately, references from this field to the existing mathematical literature on the subject are virtually nonexistent. It has been noticed before~\cite[Sec.~2.2.1.1]{KlyMarg} that this constitutes a ``disturbing example of a split between mathematics and physics''. One of our goals is to ameliorate this situation a bit by pointing out some of the literature on both sides.

Marginal problems have also arisen in the following other fields of mathematical research:

\begin{enumerate}
\item knowledge integration of expert systems in artificial intelligence~\cite{Vomlel},
\item database theory and privacy aspects of databases~\cite{Ab,Chowdhury,Dwork},
\item Vorob'ev's theory of \emph{coalition games}~\cite{Vorob2}.
\end{enumerate}

We will present a more detailed exposition of how marginal problems arise in these contexts in section~\ref{margprob}.

The origin of this subject can be traced back to at least 1955, when Bass~\cite{Bass} has considered the case of three continuous variables with given two-variable marginals. Other early works include~\cite{DallA},~\cite{Kellerer} and~\cite{Vorob}. A more abstract and general formulation in terms of $\sigma$-algebras can be found e.g.~in~\cite{HJ}. Our ``three coins'' example appears in most papers treating marginal problems~\cite[\ldots]{AbramBrand,Vorob}, sometimes more prosaically phrased~\cite{LSW,Specker}. Some further randomly selected references studying marginal problems are~\cite{AP,Malvestuto,Studeny}. Also various quantum versions of marginal problems have been considered, see e.g.~\cite{QMP}.

\subsection*{Entropic inequalities}
The most central concept in information theory is that of \emph{Shannon entropy} and its siblings like conditional entropy, mutual information and relative entropy. Its importance manifests itself not only in the widespread use of Shannon entropy within information theory itself~\cite{CG,YeungBook}, but also in related fields like biodiversity studies~\cite{Krebs}, Bayesian statistics~\cite{Jaynes}, research on collective social behavior~\cite{SP}, or additive combinatorics~\cite{Tao}. The proofs of theorems in information theory often rely on inequalities between Shannon entropies and/or derived quantities. Therefore, it is of fundamental importance to understand all the inequalities which hold between entropies of certain collections of random variables. The so-called \emph{Shannon-type inequalities}~\cite[Ch.~13]{YeungBook} are the most frequently used kind of entropic inequalities. This is the class of all those linear inequalities which can be derived from the \emph{basic inequalities}
$$
H(X) \geq 0 ,\qquad H(X|Y) = H(XY) - H(Y) \geq 0,
$$$$
I(X:Y) = H(X) + H(Y) - H(XY) \geq 0 ,
$$$$
I(X:Y|Z) = H(XZ) + H(YZ) - H(XYZ) - H(Z) \geq 0.
$$
where each symbol $X$, $Y$, $Z$ stands for a random variable or collection of random variables. These basic inequalities express non-negativity of Shannon entropy $H(X)$, conditional entropy $H(X|Y)$, mutual information $I(X:Y)$ and conditional mutual information $I(X:Y|Z)$. Many commonly used information-theoretic inequalities are Shannon-type inequalities; see e.g.~\cite{MT} and references therein for a rather general class of such inequalities and their applications.

A linear programming framework for Shannon-type entropic inequalities has been introduced in~\cite{Yeung}, including the software packakge \texttt{ITIP} which determines whether a given linear entropic inequality is a valid Shannon-type inequality or not. Further progress has been made in~\cite{Yeung3}, where it was shown that not all valid linear inequalities among entropic quantities are Shannon-type inequalities.

Occurences of entropic inequalities outside of information theory itself include applications to group theory~\cite{CY} and to Kolmogorov complexity~\cite{KC}. In this paper, we consider an application of entropic inequalities which was originally introduced, in a less general context, by Braunstein and Caves in~\cite{CHSHentropic}. In our terminology and notation, they were working with the Shannon-type entropic inequality
$$
H(A_1A_4) + H(A_2) + H(A_3) \leq H(A_1A_2) + H(A_2A_3) + H(A_3A_4) ,
\label{bcineq}
$$
which is valid for all joint distributions of the four variables. They found that this inequality can be violated by using a ``four coin'' example similar to the one above: the relevant joint distributions of variable pairs are known, so that their entropies are well-defined, and the inequality can be evaluated. The resulting violation witnesses that there cannot exist any joint distribution compatible with the given two-variable marginals. In this sense, entropic inequalities give necessary conditions for the existence of solutions to marginal problems. See proposition~\ref{eidc}.

An important advantage of the application of entropic inequalities to marginal problems is that they apply irrespectively of the number of outcomes of each variable. On the negative side, entropic inequalities are only a \emph{sufficient} criterion for the existence of a solution to a marginal problem: many marginal problems have no solution, although no violations of corresponding Shannon-type entropic inequalities exist. For a more detailed discussion of these issues, we refer to our companion paper~\cite{ent_approach}.

\subsection{Our contributions and structure of this paper.}

We start in section~\ref{margprob} by introducing marginal problems in more detail in order to set up terminology and notation. The ``three coins'' reappear as example~\ref{triangle1}.

In section~\ref{polymatroids}, we use polymatroids as a rigorous formalism for the discussion of Shannon-type inequalities and introduce \emph{partial polymatroids}, which represent marginal problems on the level of entropies.

Section~\ref{computational} explains how to use Fourier-Motzkin elimination to compute all Shannon-type entropic inequalities for a given marginal scenario, while possibly taking into account additional (conditional) independence requirements on the joint distribution like in graphical models. We offer a \texttt{MATHEMATICA}~\cite{mathematica} package which generates, given the marginal scenario, the corresponding input for the Fourier-Motzkin solver \texttt{PORTA}~\cite{porta}. Although these computations are very demanding, we have found them to be of some use in our work on quantum nonlocality~\cite{ent_approach}. We also outline an application to causal inference along the lines of~\cite{SA}. 

We analytically solve the partial polymatroid version of marginal problems for the family of $n$-cycle marginal scenarios, $\mathcal{C}_n$ with $n\in\N$, in section~\ref{ncycle}. The resulting inequalities form a single equivalence class under the action of the cyclic symmetry of $\mathcal{C}_n$. Our proof implies that there are no non-Shannon-type inequalities in any $\mathcal{C}_n$. Finally, we use these $n$-cycle inequalities to give a bound on the decay of correlations in stationary stochastic processes.

Shannon-type entropic inequalities for differential entropy are discussed in section~\ref{differential}, where we show them to not give any non-trivial constraints on the existence of solutions to marginal problems for continuous variables. 

Section~\ref{nsti} shows that non-Shannon-type entropic inequalities can be useful for detecting the non-existence of solutions to marginal problems, both for discrete and for continuous variables.

Finally, we conclude in section~\ref{conclusion} with a list of open problems.

\subsection{Notation and conventions.} All our logarithms are with respect to base $2$. In particular, we measure entropy in bits. We take $[n]=\{1,\ldots,n\}$ to be a finite index set and write $2^{[n]}$ for the set of all subsets of $[n]$. With the exception of sections~\ref{differential} and~\ref{nsti}, all random variables occurring in this paper are assumed to be discrete in such a way that their Shannon entropy converges.

\section{Marginal problems for random variables}
\label{margprob}

In this section, we introduce marginal problems as discussed, in different variants, for example in~\cite{AbramBrand,Bass,DallA,HJ,Joe,Kellerer,LSW,Malvestuto,QR,Studeny,Vorob}.

We consider a finite number of random variables $A_1,\ldots,A_n$. For any subset $S\subseteq[n]$, we also write $A_S$ for the tuple $(A_i)_{i\in S}$. In particular, $A_{[n]}=(A_1,\ldots,A_n)$ represents the joint distribution of all variables, and we stipulate $A_{\emptyset}=0$.

In many situations, one knows the distribution of $A_S$ for certain subsets $S\subseteq[n]$, but not the joint distribution of $A_{[n]}$. Sometimes, it is unclear whether a joint distribution even exists; in this case, one deals with a marginal problem.

Now if the distribution of $A_S$ is known for some $S\subseteq [n]$, then taking marginals down to a smaller subset $S'\subset S$ yields the distribution of $A_{S'}$. Therefore, the collection of sets of variables with known distribution is naturally closed under taking subsets. This motivates the following definition:

\begin{definition}[{\cite{Vorob}}]
A \emph{marginal scenario} $\mathcal{M}$ on $[n]$ is a non-empty collection
$\mathcal{M}=\{S_1,\ldots,S_{|\mathcal{M}|}\}$ of subsets $S_i\subseteq [n]$ such that
if $S\in\mathcal{M}$ and $S'\subseteq S$, then also $S'\in\mathcal{M}$.
\end{definition}

In its topological interpretation, such a combinatorial structure is also known as an \emph{abstract simplicial complex}~\cite{Munk}. 

Clearly, a marginal scenario is determined by those subsets $S_i\in\mathcal{M}$ which are not contained in any other $S_j\in\mathcal{M}$; such a subset is \emph{maximal}. In particular, it is sufficient to specify these maximal subsets when defining a particular marginal scenario. This is the approach taken in~\cite[Sec.~2.4]{AbramBrand}, where these maximal subsets are called \emph{measurement contexts}. For any set system $\mathcal{X}\subseteq 2^{[n]}$, we let $\overline{\mathcal{X}}^{\subseteq}$ denote the set system containing $\mathcal{X}$ together with all the subsets of sets in $\mathcal{X}$. It is the marginal scenario generated by $\mathcal{X}$.

We now formalize the idea of specifying a family of compatible marginal distributions for a marginal scenario $\mathcal{M}$. If $P$ is a probability distribution on some set of variables containing $S\subseteq[n]$, then we write $P_{|S}$ for the marginal distribution associated to $A_S$.

\begin{defn}
\label{mm}
A \emph{marginal model} $P^{\mathcal{M}}$ on $\mathcal{M}$ is a collection $(P^{\mathcal{M}}_S)_{S\in\mathcal{M}}$ of probability distributions $P_S^{\mathcal{M}}$ for the variables $A_S$ such that these distributions are compatible: for any pair $S,T\in\mathcal{M}$ with $T\subseteq S$, taking the marginal $P^{\mathcal{M}}_{S|T}$ of the distribution $P_S^{\mathcal{M}}$ over those variables not contained in $T$ yields precisely the given distribution $P_T^{\mathcal{M}}$,
\beq
\label{sheafcond}
P_{S|T}^{\mathcal{M}} = P_T^{\mathcal{M}}.
\eeq
\end{defn}

In particular, this compatibility condition implies that for any triple of subsets $S,S',T\in\mathcal{M}$ with $T\subseteq S,S'$, we have $P_{S|T}^{\mathcal{M}}=P_{S'|T}^{\mathcal{M}}$, as in~\cite{AbramBrand}.

The prime example of a marginal model $P^{\mathcal{M}}$ arises when starting from a joint distribution $P$ and defining the marginal models in terms of its marginals as $P_{S}^{\mathcal{M}} = P_{|S}$. However, we will see that not all marginal models can be constructed in this way. The following terminology follows the literature on quantum contextuality, e.g.~\cite[Thm.~6]{LSW}.

\begin{defn}
$P^{\mathcal{M}}$ is \emph{non-contextual} if there exists a joint distribution $P=P(a_1,\ldots,a_n)$ for all variables $A_1,\ldots,A_n$ such that its marginals coincide with the distributions occurring in the marginal model, i.e.~if $P_S^{\mathcal{M}}=P_{|S}$ for all $S\subseteq [n]$. Otherwise, $P^{\mathcal{M}}$ is \emph{contextual}.
\end{defn}

The idea behind the term ``contextual'' is that although a contextual marginal model allows no joint distribution for all variables in the conventional sense, one can easily find compatible joint distributions which depend on the subset of variables $S\subseteq [n]$. If one does this, then the joint distribution depends on the \emph{context} in which it is probed.

Under certain assumptions on $\mathcal{M}$, every marginal model is non-contextual~\cite{Vorob}. In general, this is not so, with the most elementary example being the ``triangle'':

\begin{ex}[{\cite[\ldots]{LSW,Vorob}}]
\label{triangle1}
We now formalize the ``three coins'' example from the introduction in this language. The corresponding marginal scenario is denoted by $\mathcal{M}=\mathcal{C}_3$ and consists of three variables $A_1, A_2, A_3$ where the three pairwise marginals are assumed to be given, but not the full joint distribution, so that
$$
\mathcal{C}_3 = \overline{\left\{ \{1,2\},\{1,3\},\{2,3\} \right\}}^{\subseteq} .
$$
The three variables take values in the set $\{\textit{heads},\textit{tails}\}$, such that each single variable separately has a uniformly random outcome. The two-variable distributions are
\begin{align}
\begin{split}
\label{trianglebox}
P^{\mathcal{C}_3}_{\{1,2\}}(A_1=a_1,A_2=a_2) &= \left\{ \begin{array}{cl} 1/2 & \textrm{if }\: a_1 = a_2 \\ 0 & \textrm{if }\: a_1\neq a_2 \end{array}\right.\\
P^{\mathcal{C}_3}_{\{2,3\}}(A_2=a_2,A_3=a_3) &= \left\{ \begin{array}{cl} 1/2 & \textrm{if }\: a_2 = a_3 \\ 0 & \textrm{if }\: a_2\neq a_3 \end{array}\right.\\
P^{\mathcal{C}_3}_{\{1,3\}}(A_1=a_1,A_3=a_3) &= \left\{ \begin{array}{cl} 0 & \textrm{if }\: a_1 = a_3 \\ 1/2 & \textrm{if }\: a_1\neq a_3 \end{array}\right..
\end{split}
\end{align}
These determine the single-variable distributions $P^{\mathcal{C}_3}_{ \{ 1 \} }$, $P^{\mathcal{C}_3}_{ \{ 2 \} }$, $P^{\mathcal{C}_3}_{ \{ 3 \} }$ in a consistent way, thereby satisfying the premise of definition~\ref{mm}. We claim that this marginal model is contextual. To see this, let $P$ be a hypothetical joint distribution. Then \emph{any} joint outcome probability
$$
P(A_1=a_1,A_2=a_2,A_3=a_3)
$$
needs to vanish: if $a_1\neq a_2$, this follows from the requirement $P_{|\{1,2\}}{=} P^{\mathcal{C}_3}_{\{1,2\}}$; if $a_2\neq a_3$, it follows from $P_{|\{2,3\}}{=} P^{\mathcal{C}_3}_{\{2,3\}}$; the remaining case is $a_1=a_2=a_3$, and then it is implied by $P_{|\{2,3\}}{=} P^{\mathcal{C}_3}_{\{2,3\}}$. 
\end{ex}

This example demonstrates the existence of contextual marginal models. The question now is the following:

\begin{prob}[{Marginal Problem}]
\label{MP}
How to decide whether a given marginal model $P^{\mathcal{M}}$ in a given marginal scenario $\mathcal{M}$ is non-contextual or contextual?
\end{prob}

\begin{rem}
\label{mmlp}
When the number of outcomes of each variable is finite, this is a linear programming problem: the joint distribution can be identified with its list of outcome probabilities, which are nonnegative real numbers subject to a list of equations (reproduction of the given marginals). See~\cite{AbramBrand} for an explicit formulation of this linear program. However, the number of variables in this linear program is in general exponential in $\mathcal{M}$; when the number of outcomes of each variable is $d$, then it is $d^n$, corresponding to the size of a joint distribution. In fact, certain classes of marginal problems are known to be NP-complete~\cite{pitowsky}.
\end{rem}

The entropic inequalities we are going to study in the following sections are necessary conditions for a marginal model to be non-contextual.

We end this section by discussing how marginal problems, and intimately related issues, arise in various mathematical sciences. This list is certainly not complete, but merely represents our own limited knowledge.

\begin{enumerate}
\item In quantum theory, a physical system is described by a Hilbert space $\H$, to which one associates the $C^*$-algebra $\mathcal{B}(\H)$ of bounded operators. A state is a unit vector $\psi\in\H$, while an observable is a hermitian operator $A=A^*\in\B(\H)$. If $A=\sum_i \lambda_i Q_i$ is the spectral decomposition\footnote{Since we take our random variables to be discrete, we also assume the operator $A$ to be discrete, i.e.~to have pure point spectrum.} of $A$, then the \emph{Born rule} states that the outcome distribution associated to a measurement of $A$ is given by $P(A=\lambda_i) = \langle\psi,Q_i\psi\rangle$. If two or more observables $A_1,\ldots,A_n$ are hermitian operators which commute pairwise, and have spectral decompositions $A_j = \sum_i \lambda_{j,i} Q_{j,i}$, then they are jointly measurable, and their joint distribution is given by
$$
P(A_1=\lambda_{1,i_1},\ldots,A_n=\lambda_{n,i_n}) = \langle \psi, Q_{1,i_1}\ldots Q_{n,i_n} \psi \rangle .
$$
However, if the variables are not pairwise commuting, then they cannot be jointly measured, and their joint outcome distribution is undefined. A marginal scenario $\mathcal{M}$ can then be defined as containing all those subsets $S\subseteq[n]$ for which the associated operators are pairwise commuting. The resulting outcome distributions then define a marginal model on $\mathcal{M}$. As witnessed by the Kochen-Specker theorem~\cite{KStheorem} and by Bell's theorem~\cite{Bell,Fine}, this marginal model is often contextual. This is the essence of \emph{quantum contextuality}. We refer to~\cite{AbramBrand,LSW} and our companion paper~\cite{ent_approach} for more detail and the explanation for why some of these contextual marginal models can be interpreted as \emph{quantum nonlocality}. The latter marginal models---also known as \emph{nonlocal correlations}---have been found to be a useful resource for information processing and communication tasks~\cite{Complexity,random,QKD}.

\item Knowledge integration of expert systems: many artificial intelligence systems aggregate information from several sources. In general, none of these sources will provide perfect information about the state of the world, but only about certain aspects of it; typically, the aspects probed by different sources will overlap. The system then faces the problem of \emph{integrating} the given observations into a consistent picture of the state of the world. Mathematically, this boils down to finding a probability distribution consistent with a given collection of marginals; in this context, a marginal problem asks whether such knowledge integration is possible. For more information on a popular algorithm used for finding a global joint distribution and an analysis of its behavior in the contextual case, we refer to~\cite{Vomlel}.

\item Database theory and privacy aspects of databases: this is best illustrated with an example. A health insurance provider typically has an enormous database of patients which contains, for each patient, a long list of properties like gender, age, diseases, nationality, clinical history, etc. The associated statistics of this data will be of great interest to managers, politicians, researchers in medicine and the general public. However, making the complete database available would compromise the privacy of the patients and is therefore not an option: even after discarding patient names, the database is still likely to contain enough information to make some individual entries be uniquely identifiable with certain persons. Hence there is a balance between the usefulness of the data released and the privacy of the individuals in the database. One approach for achieving such a balance lies in releasing only certain \emph{marginals} of the table~\cite{Ab}: for example, the joint distribution of gender, age, and heart disease prevalence. Given a collection of such marginals, the question is obvious: what do those marginals reveal about the database itself~\cite{Chowdhury}? This is very similar to a marginal problem and we expect some of our methods to also apply in this situation. The question of contextuality of a marginal model reappears as soon as one also adds random components to the marginals before releasing them in order to further increase privacy~\cite{Dwork}.

\item Vorob'ev's theory of coalition games: a coalition game features a finite set of players together with a collection of \emph{coalitions}, where each coalition is a subset of the players. A player may belong to any number of coalitions. Each player has a finite set of \emph{pure strategies} representing his possible actions. A \emph{mixed strategy} is a probability distribution over the set of pure strategies. The strategies chosen by the players do not have to be independent, so that the global strategy of all players is specified by a joint distribution over strategy assignments. Roughly speaking, each coalition specifies a joint mixed strategy for its players. The question then is whether there is a global mixed strategy marginalizing to those specified by the coalitions. This is a marginal problem.

We note that the standard notion of ``coalition game'' is not Vorob'ev's, but rather refers to cooperative game theory~\cite{Shapley2}.
\end{enumerate}

\section{The entropy cone and polymatroids}
\label{polymatroids}

Surprisingly, in some cases the contextuality of a marginal model can be detected already by only looking at the \emph{Shannon entropies} of the given marginal distributions. To our knowledge, this has first been noticed by Braunstein and Caves~\cite{CHSHentropic} in the case of marginal models arising from quantum nonlocality. Before getting to these ideas, we begin by recalling some properties of Shannon entropy.

\subsection{The entropy cone}
\label{subsec:entropic_cone}

Let $A_1,\ldots,A_n$ be random variables with a certain joint distribution. We do not explicitly specify the codomain of these variables, which can be any set; however, we always assume it to be finite or countable in such a way that all the entropies which we consider in the following are well-defined and finite. We write $P(a_1,\ldots,a_n)$, or simply $P$, for their joint distribution over outcome tuples $(a_1,\ldots,a_n)$. 

For any subset of indices $S\subseteq [n]$, we consider the joint Shannon entropy associated to the marginal distribution $P_{|S}$ of $A_S$,
$$
H(A_S) = - \sum_{a_S} P_{|S}(a_S) \log P_{|S}(a_S) .
$$
As a degenerate case, the distribution $P_{|\emptyset}$ is the unique probability distribution on one outcome, and hence its entropy is given by $H(A_{\emptyset})=0$.

The vector $\left(H(A_\emptyset),\ldots,H(A_{[n]})\right)$, where the components range over all of the $H(A_S)$ for $S\subseteq[n]$, is a point in $\R^{2^{[n]}}$. The collection of all points in $\R^{2^{[n]}}$ which arise from probability distributions in this way is difficult to characterize. Its closure is known to be a convex cone~\cite[Thm.~15.5]{YeungBook}, the (closed) \emph{entropy cone} $\overline{\Gamma}^*_n$~\cite{Yeung,YeungBook,Yeung3}. Since any closed convex cone can be described in terms of the linear inequalities which bound it, one may now ask: what is the description of $\overline{\Gamma}^*_n$ in terms of linear inequalities?

There are some obvious linear constraints satisfied by all points in $H(\cdot)\in\overline{\Gamma}^*_n$. For example, $H(A_S)\geq 0$ for every $S\subseteq\{1,\ldots,n\}$, and $H(A_{\emptyset})=0$. More generally, every point in $H(\cdot)\in\overline{\Gamma}^*_n$ satisfies the following \emph{basic inequalities}~\cite{Yeung},
\begin{align}
0 &\leq H(A_S) \qquad \textrm{(with equality if $S=\emptyset$)}\\
H(A_S) & \leq H(A_T) \quad\textrm{if }\: S\subseteq T \\
H(A_{S\cap T}) + H(A_{S\cup T}) &\leq H(A_S) + H(A_T)
\end{align}
for every pair of subsets $S,T\subseteq\{1,\ldots,n\}$. As already stated in the introduction, the second and third inequalities can be regarded as saying that the conditional entropy $H(A_T|A_S)$ and the conditional mutual information $I(A_S:A_T|A_{S\cap T})$ are non-negative.

\begin{defn}
A linear inequality in the $H(A_S)$'s is a \emph{Shannon-type inequality} if it is a non-negative linear combination of the basic inequalities.
\end{defn}

The collection of basic inequalities would be a complete description of $\overline{\Gamma}_n^*$ if all inequalities valid for $\overline{\Gamma}_n^*$ were Shannon-type. However, for $n\geq 4$, this is known not to be the case~\cite{Yeung3},~\cite[Thm.~15.7]{YeungBook}. As far as we know, finding the complete inequality description of $\overline{\Gamma}_n^*$ remains an elusive problem. 

The Shannon-type inequalities are the ones which are most commonly used in information theory. We will also focus on Shannon-type inequalities for the most part.

\subsection{Polymatroids}

We now would like to ask, is it possible to detect the contextuality of a marginal model by looking at the entropies of the given marginals and finding that a Shannon-type inequality is violated? Clearly, in order for such an inequality to be applicable, it should only depend on those $H(A_S)$ for which $S\in\mathcal{M}$, so that the distribution of $A_S$ is given. We would like to talk about the collection of Shannon-type inequalities which can be used in this way. This requires us to not work with the cone $\overline{\Gamma}^*_n$, but rather with the collection of all vectors in $\R^{2^{[n]}}$ which satisfy the basic inequalities. This is the convex cone of \emph{polymatroids}:

\begin{definition}[{\cite{Edmonds}}]
A \emph{polymatroid} is a pair $([n],f)$ where $n\in\N$ and $f$ is the \emph{rank function}, a function $f:2^{[n]}\to\R$ which satisfies the \emph{basic inequalities}
\begin{align}
\label{nonneg}  0 &\leq f(S) \qquad \textrm{(with equality if $S=\emptyset$)}\\
\label{increase}  f(S) & \leq f(T) \quad\textrm{if }\: S\subseteq T \\
\label{submod}  f(S\cap T) + f(S\cup T) &\leq f(S) + f(T)
\end{align}
for all $S,T\subseteq [n]$.
\end{definition}

We will usually identify a polymatroid with its rank function. The reason for introducing polymatroids lies in the fact that they are \emph{defined} via the basic inequalities: upon replacing $f(S)$ by $H(A_S)$, the linear inequalities satisfied by all polymatroids become precisely the Shannon-type entropic inequalities.

If $A_1,\ldots,A_n$ are discrete random variables with a certain joint distribution, then $f:S\mapsto H(A_S)$ is a polymatroid. A polymatroid arising in this way is called \emph{entropic}. Due to the existence of non-Shannon-type inequalities~\cite{Yeung3}, not every polymatroid is entropic. Although there are other important classes of polymatroids, like polymatroids associated to hypergraphs~\cite{VW}, network flows~\cite{Megiddo} with applications to network coding~\cite{Han}, we will always have entropy in mind. Notwithstanding, the results of this and the following two sections apply generally.

The submodularity inequality~(\ref{submod}) is a natural convexity-like condition which can be interpreted as follows. One may think of $[n]$ as a set of possible tasks which can be completed, and of $f(S)$ for $S\subseteq [n]$ as the amount of resources that have to be spent---for example, work---in order to complete all tasks in $S$. Completing the tasks $S\cup T$ is at most as difficult as completing $S$ plus completing $T$, so that $f(S\cup T)\leq f(S) + f(T)$; since having completed a task $i\in S$ may help in completing another task $j\in T$, this inequality will in general be strict. Similar considerations explain~(\ref{submod}), if one applies this argument to the additional cost relative to the tasks $S\cap T$: if the tasks $S\cap T$ are already all done, then the additional cost to complete $S\cup T$ should be less than or equal to the cost to complete $S$ plus the cost to complete $T$. This suggests
$$
f(S\cup T) - f(S\cap T) \leq \left[ f(S) - f(S\cap T) \right] + \left[ f(T) - f(S\cap T) \right] ,
$$
which is~(\ref{submod}).

Since the defining inequalities are linear, the sum of two polymatroids is again a polymatroid; similarly, a positive scalar multiple of a polymatroid is again a polymatroid. Therefore, the set of all polymatroids on $[n]$ is a convex cone denoted by $\Gamma_n\subseteq\R^{2^{[n]}}$. As already noted, we have the inclusion $\overline{\Gamma}_n^*\subseteq \Gamma_n$, which is strict for $n\geq 4$.

\begin{prop}
All basic inequalities follow from the following ones:
\begin{align}
\begin{split}
\label{shannonineqs}
f([n]\setminus\{i\}) &\leq f([n]) \qquad \forall i\in [n], \\
f(R) + f(R\cup\{i,j\}) &\leq f(R\cup \{i\}) + f(R\cup \{j\}) \qquad \forall R\subseteq [n],\: i,j\in [n]\setminus R\textrm{ with }i\neq j, \\
f(\emptyset) &= 0 .
\end{split}
\end{align}
\end{prop}

\begin{proof}
This result is well-known~\cite[Sec.~14]{YeungBook}.
\end{proof}

\subsection{Marginal problems for polymatroids}

We now define a version of marginal problems which is not about random variables, but about polymatroids. By taking entropies, a marginal problem for random variables can be mapped into a marginal problem for a polymatroid, such that contextuality of the latters implies contextuality of the former (but not conversely, in general); see figure~\ref{diagram}. The following definition introduces the polymatroid analog of a marginal model:

\begin{figure}
\centerline{\xymatrix{
\textrm{joint distribution } P \ar[rrr]^{\textrm{take entropies}}_{H(\cdot)} \ar[dd]|*\txt{take marginals} &&& \textrm{polymatroid } f \ar[dd]|*\txt{ restrict to $\mathcal{M}$} \\\\
\textrm{marginal model } P^{\mathcal{M}} \ar[rrr]^{\textrm{take entropies}}_{H(\cdot)} &&& \textrm{partial polymatroid } f^{\mathcal{M}} 
}}
\caption{Relation between the different concepts discussed in the main text. By definition, a marginal model (resp.~partial polymatroid) is non-contextual if and only if it arises from a vertical arrow.}
\label{diagram}
\end{figure}
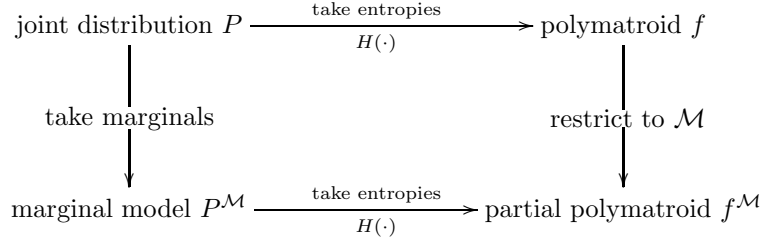

\begin{defn}
\label{defpp}
A \emph{partial polymatroid} $f^\mathcal{M}$ on a marginal scenario $\mathcal{M}$ is a function
$$
f^\mathcal{M} : \mathcal{M} \to \R
$$
which satisfies~(\ref{nonneg}),~(\ref{increase}) and~(\ref{submod}) for all $S,T\subseteq\mathcal{M}$ for which $S\cup T\in\mathcal{M}$.
\end{defn}

Intuitively, requiring the inequalities on $S,T\subseteq\mathcal{M}$ with $S\cup T\in\mathcal{M}$ only is analogous to the compatibility condition in definition~\ref{mm}. 

The most obvious example of a partial polymatroid is the restriction $f_{|\mathcal{M}}$ of a polymatroid $f: 2^{[n]}\to \R$ to $f_{|\mathcal{M}}:\mathcal{M}\to \R$. However, we will see soon that not all partial polymatroids arise in this way. These include some which come from marginal models:

\begin{prop}
Let $\mathcal{M}$ be a marginal scenario and $P^{\mathcal{M}}$ a marginal model on $\mathcal{M}$ for variables $(A_i)_{i\in[n]}$. Then the function
$$
f^{\mathcal{M}} \::\: \mathcal{M}\to \R ,\qquad S\mapsto H(A_S)
$$
is a partial polymatroid on $\mathcal{M}$.
\end{prop}

\begin{proof}
It is straightforward to check that this satisfies definition~\ref{defpp} by using the basic inequalities~(\ref{nonneg}),~(\ref{increase}),~(\ref{submod}) in combination with the assumption~(\ref{sheafcond}).
\end{proof}

\begin{defn}
A partial polymatroid $f^\mathcal{M}$ is \emph{non-contextual} if there is a polymatroid $f$ such that $f^{\mathcal{M}}(S)=f(S)$ for all $S\in\mathcal{M}$. Otherwise, $f^{\mathcal{M}}$ is \emph{contextual}.
\end{defn}

If a marginal model is non-contextual, then the associated partial polymatroid is clearly non-contextual, too. Hence, showing the contextuality of a polymatroid is one way to detect the contextuality of a marginal model: it gives a sufficient, but in general not necessary, criterion for contextuality of the marginal model. In analogy with problem~\ref{MP}, we therefore consider:

\begin{prob}[{Marginal problem for polymatroids}]
Given a partial polymatroid $f^{\mathcal{M}}$ on $\mathcal{M}$, under which conditions is it non-contextual?
\end{prob}

\begin{ex}
\label{triangle2}
Going back to example~\ref{triangle1}, we again consider the ``triangle'' or ``three coins'' marginal scenario
\beq
\label{triangleM}
\mathcal{C}_3 = \overline{\left\{\{1,2\},\{2,3\},\{1,3\}\right\}}^{\subseteq}.
\eeq
For any polymatroid $f$, the basic inequalities
\begin{align}
\begin{split}
\label{startbasic}
f(\{1,3\}) &\leq f(\{1,2,3\}) \\
f(\{1,2,3\}) + f(\{2\}) &\leq f(\{1,2\}) + f(\{2,3\})
\end{split}
\end{align}
imply directly
\beq
\label{triangleineq}
f(\{1,3\}) + f(\{2\}) \leq f(\{1,2\}) + f(\{2,3\}) \:.
\eeq
For later use, we refer to this inequality as the \emph{triangle inequality}\footnote{Note that it indeed has considerable similarity to the ordinary triangle inequality for a metric, $d(x,z)\leq d(x,y)+d(y,z)$.}. It is an inequality for the values of $f$ on the subsets in $\mathcal{C}_3$. Hence it can also be evaluated on partial polymatroids over $\mathcal{C}_3$, and such a violation witnesses the partial polymatroid's contextuality. For example, this happens for the partial polymatroid defined as
$$
f^{\mathcal{C}_3}(\{1\}) = f^{\mathcal{C}_3}(\{2\}) = f^{\mathcal{C}_3}(\{3\}) = 1 ,\qquad f^{\mathcal{C}_3}(\{1,2\}) = f^{\mathcal{C}_3}(\{2,3\}) = 1,\qquad f^{\mathcal{C}_3}(\{1,3\}) = 2 .
$$
This partial polymatroid can be interpreted as arising from the following marginal model, similar to~(\ref{trianglebox}),
\begin{align}
\begin{split}
\label{trianglePuc}
P_{\{1,2\}}^{\mathcal{C}_3}(A_1=a_1,A_2=a_2) &= \left\{ \begin{array}{cl} 1/2 & \textrm{if }\: a_1 = a_2 \\ 0 & \textrm{if }\: a_1\neq a_2 \end{array}\right.\\
P_{\{2,3\}}^{\mathcal{C}_3}(A_2=a_2,A_3=a_3) &= \left\{ \begin{array}{cl} 1/2 & \textrm{if }\: a_2 = a_3 \\ 0 & \textrm{if }\: a_2\neq a_3 \end{array}\right.\\[.2cm]
P_{\{1,3\}}^{\mathcal{C}_3}(A_1=a_1,A_3=a_3) &= \quad\: 1/4 \quad \forall a_1,a_3 .
\end{split}
\end{align}
As in~(\ref{trianglebox}), $A_1$ and $A_2$ are perfectly correlated, and likewise $A_2$ and $A_3$. But now, $A_1$ and $A_3$ are completely uncorrelated (instead of anticorrelated, as in~(\ref{trianglebox})).

With this definition, every single variable has $1$ bit of entropy, the joint distribution of $A_1$ and $A_2$ has $1$ bit of entropy, likewise for $A_1$ and $A_3$, and the joint distribution of $A_1$ and $A_3$ has $2$ bits of entropy. This realizes the partial polymatroid $f^{\mathcal{C}_3}$. Since the triangle inequality~(\ref{triangleineq}) is violated, there exists no joint distribution for all three variables marginalizing to the given ones, and~(\ref{triangleineq}) witnesses the contextuality of this marginal model.
\end{ex}

\begin{rem}
As in example~\ref{triangle1}, the reason why the marginal model in this example is contextual is that perfect correlation is transitive: if $A_1$ is perfectly correlated with $A_2$, and $A_2$ is perfectly correlated with $A_3$, then $A_1$ should also be perfectly correlated with $A_3$. We regard the triangle inequality~(\ref{triangleineq}) as one quantitative version of this intuition.

However, the entropies associated to the marginal model~(\ref{trianglebox}) do not violate~(\ref{triangleineq}): on the level of entropies, there is no difference between~(\ref{trianglebox}) and the marginal model in which $A_1$ are $A_3$ are also perfectly correlated (instead of anti-correlated), and this latter marginal model is obviously non-contextual. Entropies cannot distinguish between correlation and anti-correlation, and are generally very coarse invariants of probability distributions. From this point of view, it is quite surprising that entropic inequalities can witness the contextuality of some marginal models like~(\ref{trianglePuc}) at all. See also~\cite{ent_approach}.
\end{rem}

\section{Computations and applications}
\label{computational}

\subsection*{Fourier-Motzkin elimination.}

Determining whether a given partial polymatroid $f^{\mathcal{M}}$ is non-contextual means checking whether there exist values $f(S)$ for $S\in[n]\setminus\mathcal{M}$ which extend the given partial polymatroid $f^{\mathcal{M}}$ to a ``full'' polymatroid $f$. This requires values $f(S)$ such that all the basic inequalities~(\ref{shannonineqs}) hold not just on $\mathcal{M}$, but on all of $2^{[n]}$. This is a linear programming problem, as in Yeung's linear programming framework for Shannon-type entropic inequalities~\cite{Yeung}, and therefore can be solved in time polynomial in its size. However, the size of this linear program is $2^n$, which typically grows exponentially in the size of $\mathcal{M}$. This happens, for example, for the family of $n$-cycle scenarios $\mathcal{C}_n$ discussed in section~\ref{ncycle}, where the number of missing values $f(S)$, i.e.~the number of unknowns of the linear program, is $2^n-2n-1$. However, for variables with $d\gg 1$ outcomes, this is still significantly smaller than the $d^n$ of remark~\ref{mmlp}.

Given a polymatroid $f:2^{[n]}\to\R$, one computes the restriction $f_{|\mathcal{M}}:\mathcal{M}\to\R$ by simply forgetting the values $f(S)$ for all $S\in[n]\setminus\mathcal{M}$. Geometrically speaking, this is equivalent to projecting the point $f\in\R^{2^{[n]}}$ down to $\R^{\mathcal{M}}$ by forgetting some of the coordinates. Therefore, the set of non-contextual partial polymatroids on $\mathcal{M}$ is a projection of $\Gamma_n$ along a map $\R^{2^{[n]}}\to\R^\mathcal{M}$ which throws away some of the coordinates. In particular, this set is also a convex cone, and we denote it by $\Gamma^{\mathcal{M}}$. If an inequality description of $\Gamma^{\mathcal{M}}$ is known, then deciding the non-contextuality of a given partial polymatroid is simple: one only needs to check whether it satisfies all the inequalities defining $\Gamma^{\mathcal{M}}$. Therefore, it is very desirable to compute the inequality description of $\Gamma^{\mathcal{M}}$. For the cycle scenarios $\mathcal{M}=\mathcal{C}_n$ to be defind in section~\ref{ncycle}, we will find an analytic solution to this problem.

A natural way to determine such projections $\Gamma^{\mathcal{M}}$ would be calculate the extremal rays of $\Gamma_n$ and drop the irrelevant coordinates of these. The resulting points in $\R^{\mathcal{M}}$ generate the polyhedral cone $\Gamma^{\mathcal{M}}$. However, determining all the extremal rays of the cone $\Gamma_n$ is a very hard problem, with explicit solutions known only for $n\leq 5$~\cite{KK,Shapley1,cone5}. Hence this method is not practical.

A better way of determining $\Gamma^{\mathcal{M}}$ is to start from the inequality description~(\ref{shannonineqs}) of $\Gamma_n$ and then apply Fourier-Motzkin elimination. Fourier-Motzkin elimination~\cite{FM} is a standard method for calculating the inequality description of a projection of a polyhedral cone, given its inequality description. The correctness of the algorithm represents a proof showing that the projected cone is again polyhedral, i.e.~also has a description in terms of a finite number of linear inequalities. Fourier-Motzkin elimination has been implemented in various computational geometry software packages such as \texttt{PORTA}~\cite{porta}.

Since our objective of calculating the inequality description of $\Gamma^{\mathcal{M}}$ is a problem of precisely this form, it is straightforward in principle to apply Fourier-Motzkin elimination in order to achieve this for any given $\mathcal{M}$. For $n\leq 5$, we have successfully used the \texttt{PORTA} software in order to do so for various $\mathcal{M}$. In particular, we have verified the upcoming proposition~\ref{cycletight} for $n\leq 5$. Our \texttt{MATHEMATICA} program for generating a \texttt{PORTA} input file from the specification of $\mathcal{M}$ is available online~\cite{url}.

The contextuality of a marginal model can be detected by Shannon-type entropic inequalities if and only if the associated partial polymatroid lies in $\Gamma^{\mathcal{M}}$. Since Fourier-Motzkin elimination computes all the facet inequalities of $\Gamma^{\mathcal{M}}$, the resulting entropic inequalities are \emph{tight} in the following sense: they detect the contextuality of any marginal model whose contextuality can in principle be detected by Shannon-type entropic inequalities.

\subsection{Including (conditional) independence constraints.}

In certain applications like the one of the following subsection, or in one of those which we have considered in~\cite{ent_approach}, one has additional (conditional) independence constraints on the joint distributions $P(a_1,\ldots,a_n)$ constituting the solution space of a marginal problem. More explicitly, for disjoint subsets $R,S,T\subset [n]$, one might want to allow only those $P$ which satisfy the conditional independence relation that $A_S$ and $A_T$ are conditionally independent given $A_R$ (where $R$ might be empty), which can be written in entropic terms as
\beq
\label{condind}
I(A_S:A_T|A_R) = H(A_{R\cup S}) + H(A_{R\cup T}) - H(A_{R\cup S\cup T}) - H(A_R) \stackrel{!}{=} 0 .
\eeq
In general, one can also have several such constraints at the same time; For ease of presentation, we restrict to one such constraint, but the general case works in exactly the same way.

The set of all entropy assignments $S\mapsto H(A_S)$ for joint distributions satisfying~(\ref{condind}) is a face of the entropy cone $\overline{\Gamma}_n^*$. Again, we approximate $\overline{\Gamma}_n^*$ by the polymatroid cone $\Gamma_n\supseteq\overline{\Gamma}_n^*$, of which the polymatroid analogue of~(\ref{condind}), which is the equation
\beq
\label{pmcondind}
f(R\cup S) + f(R\cup T) - f(R\cup S\cup T) - f(R) \stackrel{!}{=} 0  ,
\eeq
also defines a face. As before, the image of this face of $\Gamma_n$ under the projection $\R^{2^{[n]}}\to\R^{\mathcal{M}}$ defines a cone in $\R^{\mathcal{M}}$, whose inequality description can be computed by Fourier-Motzkin elimination. A partial polymatroid $f^{\mathcal{M}}$ over $\mathcal{M}$ equals the restriction $f_{|\mathcal{M}}$ of a polymatroid $f$ satisfying~(\ref{pmcondind}) if and only if it lies in this cone in $\R^{\mathcal{M}}$.

In conclusion, our method allows the determination of a finite list of tight Shannon-type entropic inequalities for marginal scenarios also in the presence of (conditional) indepedendence constraints. Entropic inequalities seem especially useful to us in this kind of situations, since a (conditional) independence constraint is a linear equation on the level of entropies, so that the linear programming methods and Fourier-Motzkin elimination still apply. On the level of probabilities however, this is no longer the case, since a (conditional) independence constraint is a quadratic equation, resulting in a difficult system of linear inequalities subject to quadratic equations. Due to the relative ease of working on the level of entropies, we see the relevance of our formalism with respect to marginal problems in particular in situations where the marginal problem comes with additional (conditional) independence constraints.

\subsection{Computational results.} Computing projections of cones via Fourier-Motzkin elimination is costly. Using standard Fourier-Motzkin elimination and making use of symmetries to switch between the facet description and the extremal ray description of a polyhedral cone is practical for cone dimensions of up to $\approx 40$ for the highly symmetrical cones arising from combinatorial optimization problems~\cite{Christof}. In our case, the polymatroid cone $\Gamma_n$ has dimension $2^{[n]}-1$, so that we expect $n=5$ to be the highest number of variables for which one can calculate the facets of any interesting $\Gamma^{\mathcal{M}}$ with current methods. We have described one such successful application, to a marginal problem with additional independence constraints, in~\cite{ent_approach}, and now turn to another application for which our computations have unfortunately not terminated.

We also have not been able to terminate any attempted calculation of any $\Gamma^{\mathcal{M}}$ for $n\geq 6$ with those $\mathcal{M}$ in which we were interested. Due to this high computational complexity, analytical results like proposition~\ref{cycletight} are highly relevant also for practical computations using our approach.

\subsection{Example application: inference of common ancestors in Bayesian networks.} This subsection is based on~\cite{SA}, where Steudel and Ay derive entropic inequalities for a certain kind of causal inference. We outline now how our systematic approach to entropic inequalities could in principle extend their results.

A Bayesian network is a mathematical model for the causal dependencies between random variables. We restrict to a brief discussion and refer to~\cite{LGM} for more detail. One of the several equivalent definition is this:

\begin{defn}
Let $G=(V,E)$ be an acyclic directed graph.
\begin{enumerate}
\item For $v\in V$, the set of \emph{descendants} is $\mathrm{de}(v) = \{w\in V \:|\: (v,w)\in E\}$; the set of \emph{parents} is $\mathrm{pa}(v) = \{w\in V\:|\: (w,v)\in E\}$.
\item A \emph{Bayesian network} over $G$ consists of a discrete random variable $A_v$ for every $v\in V$, so that the $A_{v}$ have a joint distribution which satisfies the \emph{local Markov property}: $A_v$ is conditionally independent of $A_{V\setminus\mathrm{de}(v)}$ given $A_{\mathrm{pa}(v)}$, or, equivalently, the corresponding conditional mutual information vanishes,
\beq
\label{lmc}
I(A_v : A_{V\setminus\mathrm{de}(v)} | A_{\mathrm{pa}(v)} ) = 0 \quad \forall v\in V.
\eeq
\end{enumerate}
\end{defn}

Intuitively, the edges $E$ model the causal dependencies between the variables $A_v$: every $A_v$ can be regarded as a probabilistic function of $A_{\mathrm{pa}(V)}$, and this function is independent of everything else, so that no other causal dependencies exist. The simplest nontrivial example of a Bayesian network is a Markov chain $\xymatrix{A_1\ar[r] & A_2\ar[r] & A_3}$.

We now imagine a situation in which only a certain subset of variables $A_i$ for $i\in M\subseteq[n]$ can be accessed, so that their joint distribution is known; the $A_i$ for $i\not\in M$ on the other hand are ``hidden variables'' which mediate the correlations between the accessible variables via the topology of the Bayesian network, but their distribution cannot be determined. The question then is,

\begin{prob}
What can be said about the topology of the network given only the joint distribution of the accessible variables $(A_i)_{i\in M}$?
\end{prob}

Entropic inequalities give necessary conditions for a certain distribution of the $(A_i)_{i\in M}$ to come from a certain network topology. In our framework, these are the entropic inequalities corresponding to the marginal scenario $\mathcal{M}=2^M\subseteq 2^{[n]}$ with additional independence requirements given by the local Markov property~(\ref{lmc}). As explained earlier, these can be calculated by the familiar Fourier-Motzkin elimination algorithm, at least in principle.

The local Markov property~(\ref{lmc}) implies other conditional independence relations for the given variables, the \emph{global Markov conditions}~\cite{LGM}. As conditional independence relations, these are linear equations for the entropies. As such, they can be used to eliminate many of the joint entropies $H(A_S)$ for $S\subseteq 2^V$ before starting the Fourier-Motzkin elimination algorithm. Therefore, for concrete calculations of entropic inequalities for Bayesian networks, it is useful to include all global Markov conditions explicitly in order to speed up the computation, although all of these equations are implied by the local Markov property.

\begin{ex}
Consider the network topology displayed in figure~\ref{baynet}. In this case, the local Markov conditions are the following:
\begin{align}
\begin{split}
I(A_2:A_4) = 0 ,\qquad I(A_3:A_6|A_2A_4) = 0, \\
I(A_4:A_6) = 0 ,\qquad I(A_5:A_2|A_4A_6) = 0, \\
I(A_6:A_2) = 0 ,\qquad I(A_1:A_4|A_6A_2) = 0.
\end{split}
\end{align}
As indicated in figure~\ref{baynet}, we assume that the variables $A_1$, $A_3$ and $A_5$ are accessible, while $A_2$, $A_4$ and $A_6$ are hidden. The result of~\cite[Thm.~10]{SA} in this very particular special case is that this network topology implies the inequality
\beq
\label{commoninfo}
2 H(A_1 A_3 A_5) \geq H(A_1) + H(A_3) + H(A_5) .
\eeq
A slightly better condition has been derived in~\cite{FW}, which is
\beq
\label{commoninfo2}
H(A_1 A_3) + H(A_3A_5) \geq H(A_1) + H(A_3) + H(A_5) .
\eeq
This inequality actually represents a class of three inequalities equivalent to each other under cyclic permutations of the variables.

We have attempted to use Fourier-Motzkin elimination in order to calculate all (Shannon-type) entropic inequalities in this scenario and see whether this latter class of inequalities is optimal and whether there are other entropic conditions besides this one, but unfortunately our computation has not terminated. 

The interpretation of these inequalities is as follows. If some three observed variables $A_1$, $A_3$, $A_5$ arise from some causal structure in which there is no quantity influencing all three of them, then they can be modelled in terms of a network topology in which at most every pair has a common ancestor. Since, for each pair, this network of common ancestors is hidden, it can as well be subsumed into a single common ancestor variable. This gives rise to the network topology of figure~\ref{baynet}. On the other hand, if the three variables are influenced by some common variable, then their joint distribution cannot arise from a network topology as in figure~\ref{baynet}. One way to witness this is by violations of inequalities~(\ref{commoninfo}) or~(\ref{commoninfo2}); again, such a violation is a sufficient, but not a necessary condition for this kind of causal inference. The most drastic example of a violation of these inequalities occurs when all three variables are identically distributed and perfectly correlated.
\end{ex}

\begin{figure}
\begin{tikzpicture}[every node/.style=draw]
\node[circle=1pt] (A1) at (0,0) {$A_1$} ;
\node (A6) at (2,0) {$A_2$} ;
\node[circle](A2) at (3,1.73) {$A_3$} ;
\node(A4) at (2,3.46) {$A_4$} ;
\node[circle] (A3) at (0,3.46) {$A_5$} ;
\node(A5) at (-1,1.73) {$A_6$} ;
\path[->,thick] (A6) edge (A1) ;
\path[->,thick] (A6) edge (A2) ;
\path[->,thick] (A4) edge (A2) ;
\path[->,thick] (A4) edge (A3) ;
\path[->,thick] (A5) edge (A3) ;
\path[->,thick] (A5) edge (A1) ;
\end{tikzpicture}
\caption{Example Bayesian network modeling the dependency relations between six random variables. A circle represents an accessible variable, while a square stands for a hidden variable.}
\label{baynet}
\end{figure}
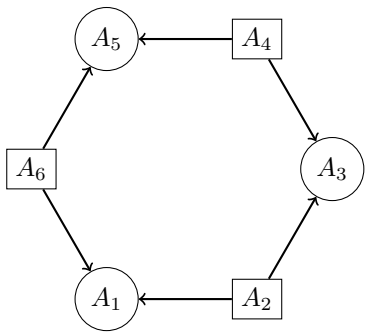

\section{Contextuality in the $n$-cycle marginal scenario}
\label{ncycle}

We now consider a family of marginal scenarios generalizing example~\ref{triangle2}. Already Vorob'ev~\cite{Vorob} (see also~\cite[Sec.~III]{LSW}) has considered the case where the marginal scenario $\mathcal{M}$ is taken to be the \emph{$n$-cycle}
\beq
\label{cyclems}
\mathcal{C}_n = \overline{\left\{ \{1,2\},\ldots,\{n-1,n\},\{n,1\} \right\}}^{\subseteq} .
\eeq
This generalizes~(\ref{triangleM}). Much more recently, the $5$-cycle has also been considered in relation to quantum contextuality~\cite{KSKlyachko}. A complete characterization of (non-)contextuality of marginal models with binary variables on $\mathcal{C}_n$ has been given in~\cite{AQB}.

In order to have somewhat more convenient notation, we regard all $i\in\N$ modulo $n$ as representatives of the elements of $[n]=\{1,\ldots,n\}$. In particular, $n+1$ and $1$ stand for the same element of $[n]$, so that we can write
$$
\mathcal{C}_n = \overline{\left\{ \{1,2\},\ldots,\{n,n+1\} \right\}}^{\subseteq} ,
$$
which will turn out to be a more useful notation for the proof below.

\begin{prop}
\label{cycletight}
Let $n\geq 3$. A partial polymatroid $f^{\mathcal{C}_n}$ is non-contextual if and only if the inequalities
\beq
\label{cycle}
f^{\mathcal{C}_n}(\{i,i+1\}) \: + \sum_{j\neq\, i,\,i+1} f^{\mathcal{C}_n}(\{j\}) \: \leq \: \sum_{j\neq i} f^{\mathcal{C}_n}(\{j,j+1\}) \qquad \forall i=1,\ldots,n
\eeq
hold.
\end{prop}

Note that all these inequalities are equivalent to each other via cyclic permutations of $[n]$. In different form, these inequalities have also been derived in~\cite{CHSHentropic}.

\begin{proof}
For ease of notation, we drop the superscript and write $f$ instead of $f^{\mathcal{C}_n}$.

If $f$ is non-contextual, we take it to be the restriction of a full polymatroid, also denoted by $f$. We now show that~(\ref{cycle}) then follows from the triangle inequality~(\ref{triangleineq}) by induction on $n$. Due to cyclic symmetry, it is sufficient to prove this in the case $i=n$. For $n=3$, the induction basis, this is precisely~(\ref{triangleineq}) itself.

For the induction step, we start with the induction assumption
$$
f(\{1,n-1\}) + \sum_{j=2}^{n-2} f(\{j\}) \leq \sum_{j=2}^{n-2} f(\{j,j+1\})
$$
to which we add the triangle inequality
$$
f(\{1,n\}) + f(\{n-1\}) \leq f(\{1,n-1\}) + f(\{n-1,n\})
$$
and get by canceling terms
$$
f(\{1,n\}) + \sum_{j=2}^{n-1} f(\{j\}) \leq \sum_{j=1}^{n-1} f(\{j,j+1\}) ,
$$
as desired.

Concerning the other implication direction, we start from a partial polymatroid $f$ defined on $\mathcal{C}_n$ satisfying~(\ref{cycle}) in addition to the basic inequalities
\begin{align}
\label{cyclenn}
0 \leq {}& f(\{i\}) , \\[6pt]
\label{cyclebasic1}
f(\{i\}) \leq f(\{i,i+1\}) ,&\qquad f(\{i+1\}) \leq f(\{i,i+1\}) , \\[6pt]
\label{cyclebasic2}
f(\{i,i+1\}) \leq {}& f(\{i\}) + f(\{i+1\}) ,
\end{align}
and prove that such an $f$ can be extended to a full polymatroid.

The inequalities~(\ref{cycle}) define a convex cone which contains $\Gamma^{\mathcal{C}_n}$. Our goal is to show that these two cones actually coincide. To this end, it is enough to consider the extremal rays of the former cone and prove that they are non-contextual as partial polymatroids. Since the inequalities~(\ref{cycle}) have integer coefficients, each extremal ray can be represented by a partial polymatroid $f$ assuming only integer values. Hence it is enough to prove the assertion for integer-valued $f$, which we assume to be the case from now on.

We now use induction on the ``total rank'' value $r_f=\sum_i f(\{i\}) + \sum_i f(\{i,i+1\})$ in order to prove the non-contextuality of $f$. The base case is $r_f=0$, which is trivial. The induction step consists in finding a non-zero polymatroid $g$ such that $f'=f-g_{|\mathcal{C}_n}$ is again a partial polymatroid satisfying the requirements~(\ref{cycle}),~(\ref{cyclenn}),~(\ref{cyclebasic1}) and~(\ref{cyclebasic2}). Since $r_{f'} < r_f$, the induction assumption applies, and $f'=h_{|\mathcal{C}_n}$ for some polymatroid $h$. Then $f=(h+g)_{|\mathcal{C}_n}$, as desired.

The $g$ we will construct will actually take values in $\{0,1\}$. In order to show that $f'=f-g$ satisfies all the desired inequalities, we always proceed as follows: if $g$ saturates a particular inequality (i.e.~satisfies it with equality), then $f'$ with also satisfy it since $f$ does so. If $g$ does not saturate it, then saturation fails by $1$, and in all these cases we will find that $f$ does not saturate the inequality either. Since $f$ is integer-valued, saturation has to fail by at least $1$. Therefore, $f'=f-g$ will also satisfy the inequality.

To find an appropriate polymatroid $g$, we distinguish four cases.

\renewcommand{\labelenumi}{\textbf{Case }\arabic{enumi}:}
\begin{enumerate}
\item $f(\{i,i+1\}) < f(\{i\}) + f(\{i+1\})$ for all $i$. In information-theoretic terms, this says that there is positive mutual information between each variable $i$ and $i+1$.

In this case, we define $g$ to be the polymatroid taking on a constant value of $1$ on all non-empty subsets of $[n]$; in particular, $g(\{i\})=g(\{i,i+1\})=1$ for all $i$. Then defining $f'=f-g$ will work: $f'$ satisfies~(\ref{cyclenn}) because the assumption implies $f(\{i\})\geq 1$ for all $i$; moreover, since $g$ saturates~(\ref{cycle}) and~(\ref{cyclebasic1}) for every $i$, the new $f'$ will satisfy these inequalities just as $f$ itself does. Finally,~(\ref{cyclebasic2}) holds for $f'$ since $f(\{i,i+1\})\leq f(\{i\})+f(\{i+1\}) - 1$ by assumption, and $g(\{i,i+1\}) = g(\{i\})+g(\{i+1\}) - 1$, so that
\begin{align*}
f'(\{i,i+1\}) &= f(\{i,i+1\}) - g(\{i,i+1\}) \\
 &\leq f(\{i\}) + f(\{i+1\}) - 1 - g(\{i\}) - g(\{i+1\}) + 1 \\
 &= f'(\{i\}) + f'(\{i+1\}) .
\end{align*}

\item There is exactly one $i$ for which $f(\{i,i+1\}) = f(\{i\}) + f(\{i+1\})$. In entropic terms, there is exactly one pair of neighboring variables which are independent. This implies that $f(\{j\})>0$ for all $j$, since otherwise $f(\{j,j+1\}) = f(\{j\}) + f(\{j+1\})$ would hold for at least two values of $j$.

We take this $i$ to be $i=n$ without loss of generality, so that
\beq
\label{1Nindep}
f(\{1,n\})=f(\{1\}) + f(\{n\})
\eeq
is assumed. Then we claim that the set of inequalities~(\ref{cycle}) is equivalent to the single inequality
\beq
\label{chain}
\sum_{i=1}^{n} f(\{i\}) \leq \sum_{i=1}^{n-1} f(\{i,i+1\}) ,
\eeq
given that~(\ref{cyclebasic1}) and~(\ref{cyclebasic2}) hold. For if $i\neq n$, we have
\begin{align*}
f(\{i,i+1\}) + \sum_{j\neq i,\, i+1} f(\{j\}) &\stackrel{(\ref{cyclebasic2})}{\leq} \sum_j f(\{j\}) \\
&\stackrel{(\ref{1Nindep})}{=} f(\{1,n\}) + \sum_{j=2}^i f(\{j\}) + \sum_{j=i+1}^{n-1} f(\{j\}) \\
&\stackrel{(\ref{cyclebasic1})}{\leq} f(\{1,n\}) + \sum_{j=2}^{i} f(\{j-1,j\}) + \sum_{j=i+1}^{n-1} f(\{j,j+1\}) \\
&= \sum_{j\neq i} f(\{j,j+1\}) ,
\end{align*}
which is precisely~(\ref{cycle}). For $i=n$ in turn,~(\ref{cycle}) coincides with~(\ref{chain}) under the assumption~(\ref{1Nindep}).

Now let $m$ be the smallest index value with the property that $f(\{m,m+1\}) > f(\{m+1\})$. Since by assumption, $f(\{2\})\leq f(\{1,2\}) < f(\{1\}) + f(\{2\})$, we have $f(\{1\})>0$. Therefore,~(\ref{chain}) implies $m\leq n-1$.

We define the polymatroid $g$ by taking it to assume the value $1$ on any set $S\subseteq [n]$ for which $S\cap\{1,\ldots,m\}\neq\emptyset$, and $0$ otherwise. This gives in particular,
$$
g(\{j\}) = \left\{\begin{array}{cl} 1 & \textrm{for } j\in\{1,\ldots,m\} \\ 0 & \textrm{otherwise} \end{array}\right. , \qquad g(\{j,j+1\}) = \left\{\begin{array}{cl} 1 & \textrm{for } j\in\{n,1,\ldots,m\} \\ 0 & \textrm{otherwise} \end{array}\right. .
$$
This $g$ corresponds to the situtation where the variables $1,\ldots,m$ have $1$ bit of entropy and are perfectly correlated, while all others are deterministic and have vanishing entropy.

It needs to be shown that setting $f'=f-g$ defines a partial polymatroid of the same kind, which means checking whether the equations~(\ref{cyclenn}),~(\ref{cyclebasic1}),~(\ref{cyclebasic2}) and~(\ref{chain}) hold. We know $f(\{j\})\geq 1$ for all  $j$, so that $f'(\{j\})\geq 0$. Furthermore, $g$ saturates $g(\{i\})\leq g(\{i,i+1\})$ for all $i$ except for $i=n$; therefore, $f'$ satisfies $f'(\{i\})\leq f'(\{i,i+1\})$ for all $i\neq n$. Moreover, since $f(\{1\})>0$, we have $f(\{1,n\}) = f(\{1\}) + f(\{n\}) \geq f(\{n\}) + 1$, and so
$$
f'(\{n\}) = f(\{n\}) - g(\{n\}) \leq f(\{1,n\}) - 1 = f'(\{n,1\}).
$$
A similar distinction of cases shows $f'(\{i+1\})\leq f'(\{i,i+1\})$ for all $i$. That~(\ref{cyclebasic2}) holds for $f'$ can be verified similarly: $g$ saturates this inequality for all $j\not\in\{1,\ldots,m-1\}$, whereas $f$ does not saturate it for the other values of $j$ by assumption; therefore, $f'$ satisfies it. Finally, $g$ also saturates~(\ref{chain}), so that $f'$ will also satisfy it since $f$ does.

\item Still $f(\{j\})>0$ for all $j$, but now there are two or more values of $i$ for which $f(\{i,i+1\}) = f(\{i\}) + f(\{i+1\})$,

As in the previous case, we take one of these values to be $i=n$. Then the same observations apply:~(\ref{cycle}) is equivalent to~(\ref{chain}). Moreover, even that inequality is now automatic: for $k\neq n$ being the smallest value for which also $f(\{k,k+1\}) = f(\{k\}) + f(\{k+1\})$, we have $k\leq n-1$ by assumption, and therefore
$$
\sum_{i=1}^{n} f(\{i\}) = \sum_{i=1}^{k-1} f(\{i\}) + f(\{k,k+1\}) + \sum_{i=k+2}^{n} f(\{i\}) \stackrel{(\ref{cyclebasic1})}{\leq} \sum_{i=1}^{n-1} f(\{i,i+1\}) .
$$
As in the previous case, we define $g$ by setting $g(S)$ to be $1$ if $S\cap\{1,\ldots,k\}\neq\emptyset$, while $g(S)=0$ otherwise. This means in particular,
$$
g(\{j\}) = \left\{\begin{array}{cl} 1 & \textrm{for } j\in\{1,\ldots,k\} \\ 0 & \textrm{otherwise} \end{array}\right. , \qquad g(\{j,j+1\}) = \left\{\begin{array}{cl} 1 & \textrm{for } j\in\{n,1,\ldots,k\} \\ 0 & \textrm{otherwise} \end{array}\right. ,
$$
Defining $f'=f-g$ now gives the desired new partial polymatroid: by the observation of the previous paragraph, it is enough to check~(\ref{cyclebasic1}) and~(\ref{cyclebasic2}), and then~(\ref{cycle}) will be automatic. Checking this can be done as at the end of the previous case.

\item $f(\{j\}) = 0$ for some $j$. Thanks to the cyclic symmety, we may consider the case $f(\{n\})=0$ and $f(\{1\})>0$ without loss of generality. Then let $k$ be the smallest index for which $f(\{k,k+1\})>f(\{k+1\})$. In particular, $f(\{j\})>0$ for all $j\in\{1,\ldots,k\}$. The polymatroid $g$ can be defined as in the previous case. Then $f'=f-g$ clearly satisfies~(\ref{cyclenn}). It also satisfies~(\ref{cyclebasic1}), since $g$ saturates these except for $g(\{n\}) < g(\{n,1\})$ and $g(\{k+1\}) < g(\{k,k+1\})$, which is fine since $f$ does not saturate them. Similarly for~(\ref{cyclebasic2}), which $f$ does not saturate for $i\in \{1,\ldots,k-1\}$ thanks to $f(\{i+1\}) = f(\{i,i+1\}) < f(\{i\}) + f(\{i+1\})$ and $g$ saturates for all other values of $i$. Finally, $f'(\{n\})=0$ guarantees that $f'$ also satisfies~(\ref{cycle}).
\end{enumerate}

This ends the proof.
\end{proof}

Writing the inequalities~(\ref{cycle}) in terms of entropies of random variables reads
\beq
\label{jecycle} H(A_i A_{i+1}) \: + \sum_{j\neq\, i,\,i+1} H(A_j) \: \leq \: \sum_{j\neq i} H(A_j A_{j+1}) \qquad \forall i=1,\ldots,n
\eeq
This inequality can be applied to marginal problems on $\mathcal{C}_n$:

\begin{prop}
\label{eidc}
Let $n\geq 3$. There is a marginal model on $\mathcal{C}_n$ whose contextuality gets detected by~(\ref{jecycle}).
\end{prop}

\begin{proof}
An example is given by a generalization of~(\ref{trianglePuc}) to all $n$, again with $a_i\in\{\textit{heads},\textit{tails}\}$,
\begin{align*}
P^{\mathcal{C}_n}_{\{i,i+1\}}(A_i=a_i,A_{i+1}=a_{i+1}) &= \left\{ \begin{array}{cl} 1/2 & \textrm{if }\: a_i = a_{i+1} \\ 0 & \textrm{if }\: a_i\neq a_{i+1} \end{array}\right. \quad\forall i=1,\ldots,n-1,\\[6pt]
P^{\mathcal{C}_n}_{\{1,n\}}(A_1=a_1,A_n=a_n) &= \quad\: 1/4 \quad \forall a_1,a_n
\end{align*}
The entropies associated to this marginal model violate~(\ref{jecycle}).
\end{proof}

See~\cite{ent_approach} for more examples, including many arising from quantum theory.

Given an integer-valued partial polymatroid $f^{\mathcal{C}_n}$ satisfying~(\ref{cycle}), the polymatroids $g$ used in the proof of proposition~\ref{cycletight} are actually entropic, so that $f$ turns out to be a sum of entropic polymatroids, and therefore is itself an entropic polymatroid. Hence we have also proven that non-Shannon type entropic inequalities cannot be relevant for marginal problems on $\mathcal{C}_n$. In other words:

\begin{cor}
\label{nononShannon}
Let $n\geq 3$. Every entropic inequality containing only terms $H(A_j)$ and $H(A_jA_{j+1})$, $j\in[n]$, is Shannon-type.
\end{cor}

\subsection{Application to correlations in stochastic processes}
\label{stocproc}
In terms of mutual information, the inequality~(\ref{jecycle}) in the case $i=n$ can be rewritten as
\beq
\label{micycle}
I(A_1:A_n) \geq \sum_{j=1}^{n-1} I(A_j:A_{j+1}) - \sum_{j=2}^{n-1} H(A_j) .
\eeq
This can be interpreted as a lower bound on the correlation between $A_1$ and $A_n$, given that there are certain correlations between each $A_i$ and $A_{i+1}$ for $i=1,\ldots,n-1$. This in particular suggests an application to stochastic processes, an idea which we briefly explore now. Let $(A_i)_{i\in\Z}$ be a stationary stochastic process. Stationarity implies that $H(A_j) = H(A_1)$ and $I(A_j:A_{j+1}) = I(A_1:A_2)$ for all $j\in\Z$. Using this, the inequality can also be written as
\beq
\label{lbcorr}
I(A_1 : A_n) \geq H(A_1) - (n-1) H(A_2|A_1) .
\eeq
We think of this as follows: let $A_1$ be a signal which undergoes $n-1$ applications of some noise, which results in noisy signals $A_2,\ldots,A_n$. These noise applications do not have to be independent, but we require them to not depend on the particular iteration and to not change the distribution of the signal in order for the resulting stochastic process to be stationary. We would like to know how well the final signal $A_n$ approximates the original signal $A_1$. Our inequality~(\ref{lbcorr}) gives a lower bound on the quality with which the original signal can be recovered from the final noisy one. The results of this section also show that this is the best linear inequality between entropies in this context. Applying the bound only requires the entropy of the signal $A_1$ to be known together with $H(A_2|A_1)$, which quantifies the amount of noise added at each timestep. One consequence is that the noisy signal $A_n$ contains some information about the original signal $A_1$ for all $n\leq\left\lfloor\frac{H(A_1)}{H(A_2|A_1)}\right\rfloor$.

\section{Shannon-type inequalities for differential entropy}
\label{differential}

It is an appealing feature of entropic inequalities that they apply regardless of the number of outcomes of each variable. One may think that this makes entropic inequalities also applicable to random variables having an uncountable range, like random variables described by continuous probability density functions, if one replaces the discrete entropy $-\sum_i p_i\log p_i$ by the differential entropy $-\int p(x)\log p(x) \, dx$. Unfortunately, this turns out not to be the case. The problem with this approach is that, for given continuous random variables $A_1,\ldots,A_n$, the joint differential entropies do not form a polymatroid. Although submodularity~(\ref{submod}) remains valid~\cite[(10.136)]{YeungBook}, monotonicity fails: for example, when $A_1$ and $A_2$ are independent and uniformly distributed on $[0,\varepsilon]$ for some $\varepsilon>0$, then the differential entropies are
$$
h(A_1) = h(A_2) = \log\varepsilon \:, \qquad h(A_1 A_2) = 2\log\varepsilon
$$
So for $\varepsilon < 1$, we have $h(A_1 A_2) < h(A_1) < 0$. Intuitively, the reason for this is that differential entropy quantifies the randomness of a distribution \emph{relative} to the Lebesgue measure. In particular, this relative entropy can become negative, meaning that $h(A)$ itself can become negative. Similar considerations apply to differential conditional entropy $h(A_2|A_1) = h(A_1 A_2) - h(A_1)$.

We have seen in the previous sections that entropic inequalities can detect the contextuality of marginals models for discrete random variables. Does this also apply to marginal problems for continuous random variables if one uses differential entropy? This would be interesting, since marginal problems for continuous variables are an important topic with relevance to applied statistics~\cite{DallA,Joe,QR}. We will show in this section that this is not the case with Shannon-type inequalities, but give an example in the next section of a non-Shannon-type inequality which can detect the contextuality of a continuous-variable marginal model.

The only basic inequalities which remain valid for differential entropy are the submodularity inequalities~(\ref{submod}). Therefore, the Shannon-type inequalities for differential entropy are those which are linear combinations of submodularity inequalities only; these coincide with the \emph{balanced} entropic inequalities of Chan~\cite{Chan}.

Therefore, instead of using partial polymatroids, we now work with \emph{submodular} functions $f:2^{[n]}\to\R$. These are those functions which satisfy~(\ref{submod}) and $f(\emptyset)=0$, but not necessarily~(\ref{nonneg}) or~(\ref{increase}). Similarly, a function $f^{\mathcal{M}}:\mathcal{M}\to\R$ is called \emph{submodular} if $f^{\mathcal{M}}(\emptyset)=0$ and $f$ satisfies~(\ref{submod}) for those $S,T\in\mathcal{M}$ for which $S\cup T\in\mathcal{M}$.

In this way, the Shannon-type entropic inequalities (for differential entropy) are precisely those inequalities which hold for all submodular functions.

\begin{prop}
\label{submodnonc}
Let $f^{\mathcal{M}}:\mathcal{M}\to\R$ be submodular. Then there is a submodular function $f:2^{[n]}\to\R$ such that $f_{|\mathcal{M}} = f^{\mathcal{M}}$.
\end{prop}

\begin{proof}
We choose a set $V\subseteq [n]$ such that $V\not\in\mathcal{M}$, but such that all proper subsets of $V$ are in $\mathcal{M}$, and define $\mathcal{M}'=\mathcal{M}\cup\{V\}$. Then
$$
f^{\mathcal{M}'} \::\: \mathcal{M}' \to \R ,\qquad U \mapsto \left\{\begin{array}{cl} f^{\mathcal{M}}(U) & \textrm{if }\: U\in\mathcal{M} \\ \min_{S,T\subsetneq V} \left(f^{\mathcal{M}}(S) + f^{\mathcal{M}}(T) - f^{\mathcal{M}}(S\cap T) \right) & \textrm{if }\: U = V \end{array}\right.
$$
extends $f^{\mathcal{M}}$ from $\mathcal{M}$ to $\mathcal{M}'$. We claim that $f^{\mathcal{M}'}$ is submodular on $\mathcal{M}'$. It needs to be shown that for any $S,T\in\mathcal{M}'$ with $S\cup T\in \mathcal{M}'$, the submodularity inequality
$$
f^{\mathcal{M}'}(S\cup T) + f^{\mathcal{M}'}(S\cap T) \leq f^{\mathcal{M}'}(S) + f^{\mathcal{M}'}(T)
$$
holds. If $S\cup T\in\mathcal{M}$, this is true since $f^{\mathcal{M}'}$ restricts to the submodular $f^{\mathcal{M}}$ on $\mathcal{M}$. If $S\cup T\not\in\mathcal{M}$, then necessarily $S\cup T=V$, and the assertion follows from the definition of $f^{\mathcal{M}'}(V)$.

Repeated application of this extension procedure eventually produces a submodular extension of $f^{\mathcal{M}}$ to all subsets of $[n]$.
\end{proof}

This implies that there are no Shannon-type inequalities for differential entropy which would be able to detect contextuality of marginal models: the differential entropy function $f^{\mathcal{M}} : \mathcal{M}\to\R$ associated to a continuous-variable marginal model has an extension to a submodular function $2^{[n]}\to\R$. In particular, it satisfies all inequalities which hold for submodular functions, and these are precisely the Shannon-type inequalities for differential entropy.

\section{Non-Shannon-type inequalities}
\label{nsti}

In this section, we enumerate the variables $A_i$ by indices $i\in\{w,x,y,z\}$ instead of $i\in\{1,2,3,4\}$, as the latter choice might cause confusion.

It has been known since 1998~\cite{Yeung3} that the inclusion $\overline{\Gamma}_n^*\subseteq\Gamma_n$ is strict for $n\geq 4$. The inequality
\begin{align}
\label{nst}
\begin{split}
-4H(A_w & A_y A_z) - H(A_x A_y A_z) - H(A_w A_x) + 3 H(A_w A_y) + 3 H(A_w A_z) \\
&+ H(A_x A_y) + H(A_x A_z) + 3 H(A_y A_z) - H(A_w) - 2 H(A_y) - 2 H(A_z) \geq 0	
\end{split}
\end{align}
bounds $\overline{\Gamma}_4^*$, but not $\Gamma_4$~\cite[Thm.~15.7]{YeungBook}. We can consider this inequality as an entropic inequality in the following marginal scenario, named after the authors of~\cite{Yeung3},
\beq
\label{MZY}
\mathcal{M}_{ZY} = \overline{\left\{\{w,y,z\},\{x,y,z\},\{w,x\} \right\}}^{\subseteq} ,
\eeq
$\mathcal{M}_{ZY}$ as a simplicial complex is illustrated in figure~\ref{ZYcomplex}.

We now consider the partial polymatroid $f^{ZY}$ depicted in figure~\ref{ZYpp}. In formulas, its values are
\begin{align*}
f^{ZY}(\emptyset) = 0,\qquad & f^{ZY}(\{w\}) = f^{ZY}(\{x\}) = f^{ZY}(\{y\}) = 2, \\[6pt]
f^{ZY}(\{w,x\}) = 4,\qquad f^{ZY}(\{w,y\}) = & f^{ZY}(\{w,z\}) = f^{ZY}(\{x,y\}) = f^{ZY}(\{x,z\}) = f^{ZY}(\{y,z\}) = 3,\\[6pt]
f^{ZY}(\{w,y,& z\}) = f^{ZY}(\{x,y,z\}) = 4 .
\end{align*}
It violates (the polymatroid analogue of) inequality~(\ref{nst}).

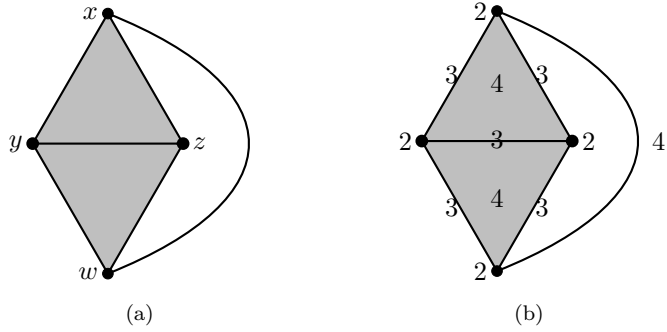
\begin{figure}
\subfigure[]{\label{ZYcomplex}
\begin{tikzpicture}
\draw[fill=lightgray,thick] (5,0) node[anchor=east] {$w$} -- (6,1.73) node[anchor=west] {$z$} -- (5,3.46) node[anchor=east] {$x$} -- (4,1.73) node[anchor=east] {$y$} -- (5,0) ;
\draw[fill=black] (5,0) circle (2pt) ;
\draw[fill=black] (5,3.46) circle (2pt) ;
\draw[fill=black,thick] (4,1.73) circle (2pt) -- (6,1.73) circle (2pt) ;
\draw[thick] (5,0) .. controls (7.5,1) and (7.5,2.46) .. (5,3.46) ;
\end{tikzpicture}}\hspace{1cm}
\subfigure[]{\label{ZYpp}
\begin{tikzpicture}
\draw[fill=lightgray,thick] (5,0) node[anchor=east] {$2$} -- (6,1.73) node[anchor=west] {$2$} -- (5,3.46) node[anchor=east] {$2$} -- (4,1.73) node[anchor=east] {$2$} -- (5,0) ;
\draw[fill=black] (5,0) circle (2pt) ;
\draw[fill=black] (5,3.46) circle (2pt) ;
\draw[fill=black,thick] (4,1.73) circle (2pt) -- (6,1.73) circle (2pt) ;
\draw[thick] (5,0) .. controls (7.5,1) and (7.5,2.46) .. (5,3.46) ;
\node at (5,0.97) {$4$} ;
\node at (5,2.50) {$4$} ;
\node at (7.15,1.73) {$4$} ;
\node at (5,1.75) {$3$} ;
\node at (4.40,2.62) {$3$} ;
\node at (5.60,2.62) {$3$} ;
\node at (4.40,0.85) {$3$} ;
\node at (5.60,0.85) {$3$} ;
\end{tikzpicture}}
\caption{\subref{ZYcomplex} The marginal scenario~(\ref{MZY}) as a simplicial complex; the numbers are vertex indices.~\subref{ZYpp} The partial polymatroid $f^{ZY}$; the numbers are the values that $f^{ZY}$ assigns to the simplices.}
\end{figure}

\begin{lem}
\label{fZY}
$f^{ZY}$ arises from taking entropies of a marginal model in the marginal scenario $\mathcal{M}_{ZY}$.
\end{lem}

\begin{proof}
We start by defining a joint distribution for $A_w$, $A_y$ and $A_z$. Let $(\alpha_w,\alpha_y,\alpha_z,\beta)$ be a list of four independent and uniformly distributed bits. Then the definitions
$$
A_w = (\alpha_w,\beta),\quad A_y = (\alpha_y,\beta),\quad A_z = (\alpha_z,\beta)
$$
reproduce the desired entropy values (in bits) for all $A_S$ with $S\subseteq\{w,y,z\}$.

Analogous definitions with $A_x$ in place of $A_w$ also define a distribution for $A_{ \{x,y,z\} }$ which reproduces the desired entropy values and restricts to the same marginal distribution of $A_{\{y,z\}}$ as the distribution of $A_{\{w,y,z\}}$. The product distribution between $A_w$ and $A_x$ defines a distribution of $A_{\{w,x\}}$ having the desired properties. This completes the definition of the marginal model.
\end{proof}

\begin{cor}
\label{nstcor}
The non-Shannon-type inequality~(\ref{nst}) detects the contextuality of some marginal models in the marginal scenario $\mathcal{M}_{ZY}$, although no Shannon-type inequality does so.
\end{cor}

\begin{proof}
The first part is clear by the lemma since $f^{ZY}$ violates~(\ref{nst}). For the second part, we have used the methods and software described in section~\ref{computational} to compute all the facet inequalities of $\Gamma^{\mathcal{M}_{ZY}}$. Since the partial polymatroid $f^{ZY}$ has turned out to violate none of these $67$ inequalities, it is non-contextual as a partial polymatroid.
\end{proof}

In particular, this shows that~(\ref{nst}) is indeed a non-Shannon-type inequality.

We now extend corollary~\ref{nstcor} to the continuous-variable case. The original proof of~(\ref{nst}) from~\cite{ZY} also works in the continuous-variable case, so that
\begin{align}
\label{nstd}
\begin{split}
-4h(A_w & A_y A_z) - h(A_x A_y A_z) - h(A_w A_x) + 3 h(A_w A_y) + 3 h(A_w A_z) \\
&+ h(A_x A_y) + h(A_x A_z) + 3 h(A_y A_z) - h(A_w) - 2 h(A_y) - 2 h(A_z) \geq 0 
\end{split}
\end{align}
is a valid non-Shannon-type inequality for differential entropy. Alternatively, this inequality can also be deduced from the results of Chan~\cite[Thm. 2]{Chan} on the relation between entropic inequalities for discrete and continuous variables.

We now claim that the partial polymatroid $f^{ZY}$ can also be realized as the collection of entropies of a continuous-variable marginal model. In order to do so, the bits in the proof of lemma~\ref{fZY} should be replaced by independent copies of a continuous variable with uniform distribution on $[0,2]$ (which has a differential entropy of $1$). This yields the desired continous-variable marginal model on $\mathcal{M}_{ZY}$. Thanks to proposition~\ref{submodnonc}, we know that no Shannon-type inequality for differential entropy can detect its contextuality, although~(\ref{nstd}) does.

\begin{rem}
The Fourier-Motzkin elimination approach of section~\ref{computational} can easily be amended so as to deal with some non-Shannon-type inequalities, too. Instead of only using the basic inequalities as the initial input to the Fourier-Motzkin solver, one can additionally provide a finite list of non-Shannon-type inequalities to begin with, and the Fourier-Motzkin solver will then also take those into account while deriving entropic inequalities applicable to a marginal scenario. However, for the practical computations that we have done~\cite{ent_approach}, this has not improved the results.
\end{rem}

\section{Open problems}
\label{conclusion}

We would now like to mention some relevant questions, which are, to the best of our knowledge, still open.

Partial polymatroids:

\renewcommand{\labelenumi}{\arabic{enumi}.}
\begin{enumerate}
\item Vorob'ev~\cite{Vorob} has found a complete characterization of those marginal scenarios in which all marginal models are non-contextual. The analogous question in the polymatroid case is this: how can one characterize the class of all marginal scenarios in which all partial polymatroids are non-contextual? Is the answer the same as in~\cite{Vorob}?
\item Proposition~\ref{cycletight} implies that for the marginals scenarios $\mathcal{C}_n$, recognizing (non-)contextuality of a partial polymatroid can be done in polynomial time. What about the complexity of this problem for other families of marginal scenarios? Is the general case even in $NP$?
\item Under which conditions on $\mathcal{M}$ can every non-contextual polymatroid be realized as an entropic polymatroid? In other words, for which $\mathcal{M}$ can one show that all entropic inequalities are Shannon-type, like in corollary~\ref{nononShannon}? For example, does this also hold when $\mathcal{M}$ is a complete graph instead of a cycle graph? If so, then this would mean that there are no non-Shannon-type entropic inequalities in which each term is of the form $H(A_i)$ or $H(A_iA_j)$.
\end{enumerate}

Entropic inequalities and marginal problems:

\begin{enumerate}[resume]
\item We have seen that taking entropies can turn contextual marginal models into contextual partial polymatroids. Which contextual polymatroids can arise in this way?
\item Upon fixing a finite set of possible outcomes for each variable, the marginal models in a given marginal scenario form a convex polytope. All the examples we have found so far~\cite{ent_approach} have the property that taking entropies of an extreme point of this polytope maps it to a non-contextual partial polymatroid. Is this always the case?
\item When writing the entropic inequalities~(\ref{jecycle}) in terms of mutual information, the resulting inequalities
$$
\sum_{j\neq i} I(A_j:A_{j+1}) - I(A_i:A_{i+1}) \leq \sum_{j\neq i,i+1} H(A_j)
$$
bear a great similarity to the inequalities derived in~\cite{AQ}, which detect contextuality in the $n$-cycle scenario for binary random variables on the level of probabilities rather than entropies. Does this similarity extend to other marginal scenarios? If so, this would provide an interesting alternative to the computationally costly Fourier-Motzkin elimination for generating Shannon-type entropic inequalities detecting the contextuality of partial polymatroids.
\end{enumerate}

\bibliographystyle{plain}
\bibliography{entropic_inequalities}

\end{document}